\documentclass[aoas,preprint,autosecdot]{imsart}  
\usepackage[T1]{fontenc}
\usepackage{amsmath}
\usepackage{amssymb}
\usepackage{amsthm}
\usepackage{wasysym}
\usepackage{esint}
\usepackage{graphicx,color}
\usepackage{mathtools}
 \usepackage{cite} 
\usepackage{url}

\newtheorem{theorem}{\bf Theorem}
\newtheorem{lemma}[theorem]{\bf Lemma}
\newtheorem{corollary}[theorem]{\bf Corollary}

\def\thmref#1{Theorem~\ref{#1}}
\def\lemref#1{Lemma~\ref{#1}}

\def\bfmath#1{\mathchoice
        {\mbox{\boldmath$#1$}}%
        {\mbox{\boldmath$#1$}}%
        {\mbox{\boldmath$\scriptstyle#1$}}%
        {\mbox{\boldmath$\scriptscriptstyle#1$}}}%

\def\bfw{\mathbf{w}}

\def\bfV{\mathbf{V}}
\def\bfW{\mathbf{W}}

\def\bbE{\mathbb{E}}

\def\cF{{\mathcal F}}

\def\cM{{\mathcal M}}

\def\cX{{\mathcal X}}
\def\cY{{\mathcal Y}}

\newcommand{\fref}[1]{Figure~\ref{#1}}
\newcommand{\sref}[1]{Section~\ref{#1}}

\def\eqdef{\stackrel{\text{\tiny def}}{=}}
\def\fc{c} 
\def\bffc{\mathbf{\fc}}
\def\ip#1{\langle #1\rangle}

\def\tb{\tau_B} 
\def\gb{g_B} 
\def\ta{\tau_A}
\def\sfs{\xi}
\def\bfsfs{\bfmath{\xi}}
\def\sfshat{\hat{\xi}}
\def\vep{\varepsilon}
\def\II{I}
\def\tII{t_{\II}}
\def\etalower{\eta_{\ell}}
\def\etaupper{\eta_{u}}
\def\Sk#1#2{S^{(#1)}_{#2}}
\def\norm#1#2{\left\Vert #1 \right\Vert_{#2}}

\def\Nfix{N_{\text{\scriptsize fix}}}
\def\subfamilyIJ{\cF_{\II,J}}
\def\subfamily{\cF}

\begin{document}
\begin{frontmatter}
\title{Fundamental limits on the accuracy of demographic inference based on\\ the sample frequency spectrum}
\runtitle{Fundamental limits on demographic inference}

\begin{aug}
\author{\fnms{Jonathan} \snm{Terhorst}\thanksref{t1}\ead[label=e1]{terhorst@stat.berkeley.edu}}
\and
\author{\fnms{Yun S.} \snm{Song}\thanksref{t2}\ead[label=e2]{yss@stat.berkeley.edu}}

\thankstext{t1}{Supported in part by a Citadel Fellowship.}
\thankstext{t2}{Supported in part by an NIH grant R01-GM109454, a Packard Fellowship for Science and Engineering, and a Miller Research Professorship.}

\runauthor{Terhorst and Song}

\affiliation{University of California, Berkeley}

\address{J. Terhorst\\
Department of Statistics\\
University of California, Berkeley\\
Berkeley, CA 94720\\
USA\\
\printead{e1}}

\address{Y.~S. Song\\
Department of Statistics and\\
\hspace{3mm} Computer Science Division\\
University of California, Berkeley\\
Berkeley, CA 94720\\
USA\\
\printead{e2}}  
\end{aug}

\begin{abstract}
The sample frequency spectrum (SFS) of DNA sequences from a collection of individuals is a summary statistic  which is
commonly used for parametric inference in population genetics.
Despite the popularity of SFS-based inference methods, currently little is known about the information-theoretic limit on the estimation accuracy as a function of sample size.
Here, we show that using the SFS to estimate the size history of a population has a
minimax error of at least $O(1/\log s)$, where $s$ is the number of
independent segregating sites used in the analysis. 
This rate is exponentially worse than known convergence rates for many classical estimation problems in statistics.
Another surprising aspect of our theoretical bound is that it does not depend on the dimension of the SFS, which is related to the number of sampled individuals.
This means that, for a fixed number $s$ of segregating sites considered, using more individuals does not help to reduce the minimax error bound.
Our result pertains to populations that have experienced a bottleneck, 
and we argue that it can be expected to apply to many populations in nature. 
\end{abstract}

\end{frontmatter}
\section{Introduction}

The past decade has seen a revolution in our ability to interrogate
the genome at the molecular level.  Fueled by technological advances
in DNA sequencing, studies now routinely query thousands or tens of thousands of 
individuals \cite{gpc2010map,nelson2012abundance,tennessen2012evolution,uk10k,exac} in order to better understand disease
susceptibility, heritability, population history, and other phenomena.
In most cases, the conclusions of these studies come in the form of
statistical estimates obtained from models that relate the effect
of interest to mutation patterns arising in sampled DNA sequences. As genetic
sample sizes explode, it is natural to wonder how additional data
improve the quality of these estimates. While this general question has
received intense focus in theoretical statistics, certain aspects of the
genetics setting (for example, non-Gaussianity and lack of independence
among samples) complicate efforts to study such models using classical
techniques. New methods are needed to theoretically characterize 
some common models in statistical genetics.

Here, we address this need for a specific estimation problem in
population genetics known as \emph{demographic inference}. As we explain
in further detail below, the aim of this problem is to reconstruct the
sequence of historical events---including population size changes, migration,
and admixture---that gave rise to present-day populations, using
DNA samples obtained from those populations.  We focus on the
simplest problem of estimating the size history of a single population
backwards in time. 

A summary statistic known as the \emph{sample frequency spectrum} (SFS; defined below)
is often employed in empirical studies \cite{nielsen2000estimation,
gutenkunst2009inferring,coventry2010deep,gazave2014neutral,gravel2011demographic, 
nelson2012abundance,excoffier2013robust,bhaskar2015efficient}, but
there have been fewer attempts to understand SFS-based estimation from a
theoretical perspective.  The main result of this paper is to show that, for
a common class of estimators which analyze the SFS, there is a
fundamental limit on their accuracy as a function of the sample size.
More precisely, we show that, under a standard statistical error metric
known as \emph{minimax error}, the rate at which these estimators
converge to the truth for certain populations is at best inversely
logarithmic in the number of independent segregating sites analyzed, and
does not depend at all on the number of individuals sampled.
Compared to other types of statistical estimation problems (for example,
linear regression), this is an extremely slow rate of convergence. Our
proof is information-theoretic in nature and applies to \emph{any}
estimator that operates solely on the SFS. This
is the first result we are aware of that characterizes the convergence
rate of demographic history estimates as a function of sample size.

The remainder of this paper is organized as follows. In
\sref{sec:prelim} we formally define our notation and model.
In \sref{sec:main} we state our main theoretical results, followed by
a discussion of their practical implications in \sref{sec:discuss}. To
streamline our exposition, all mathematical proofs are deferred until
\sref{sec:proofs}.

\section{\label{sec:prelim}Preliminaries}

The stochastic process underlying the inference
procedure we consider is Kingman's coalescent
\cite{kingman1982coalescent,kingman1982genealogy,kingman1982exchangeability}, 
which evolves backward in time and describes the genealogy of a
collection of chromosomes randomly sampled from a population. The
population size is assumed to change deterministically over time and is
described by a function $\eta:[0,\infty)\to(0,\infty)$, with $\eta(t)$
being the population size at time $t$ in the past. The instantaneous
rate of coalescence between any pair of lineages at time $t$ is
$1/\eta(t)$.

As in the standard infinite-sites model of mutation
\cite{kimura1969number}, we assume that every dimorphic site (i.e., a
site with exactly two observed allelic types) has experienced mutation
exactly once in the evolutionary history relating the sample. Further,
for each such site, we assume that it is known which allele is the
ancestral type versus the mutant type. In what follows, we use the terms
dimorphic and segregating interchangeably.

A population size function $\eta(t)$ induces a probability distribution
on the number of derived alleles found at a particular segregating
site.  Specifically, for a sample of $n \geq 2$ randomly sampled individuals, 
let $\sfs_{n,b}^{(\eta)}$, for $1\leq b\leq n-1$, denote the probability
that a segregating site contains $b$ mutant alleles in a sample of
$n$ individuals under model $\eta$.  
The vector $\bfsfs_n^{(\eta)} \eqdef (\sfs_{n,1}^{(\eta)},\ldots,\sfs_{n,n-1}^{(\eta)})$ 
is called the expected SFS.
In the coalescent setting, a general expression for $\sfs_{n,b}^{(\eta)}$ 
is given by \cite{griffiths1998age}
\[
\sfs_{n,b}^{(\eta)}\propto\sum_{k=2}^{n-b+1}\frac{\binom{n-b-1}{k-2}}{{n-1\choose k-1}} \cdot k\cdot\mathbb{E}T_{n,k}^{(\eta)},
\]
where $\mathbb{E}T_{n,k}^{(\eta)}$ denotes the amount of time (in
coalescent units) during which the genealogy of the sample contained
$k$ lineages under model $\eta$.  
The expected waiting time $\mathbb{E}T_{m,m}^{(\eta)}$ to the first coalescence 
in a sample of $m$ individuals is given by
\begin{equation}
\fc^{(\eta)}_m \eqdef \mathbb{E}T_{m,m}^{(\eta)} =\int_{0}^{\infty}t\frac{a_{m}}{\eta(t)}\exp\left\{ -a_{m}R_{\eta}(t)\right\} \,{\rm d}t,
\label{eq:ETmm}
\end{equation}
where $a_{m}\eqdef\binom{m}{2}$ and $R_{\eta}(t)\eqdef\int_{0}^{t}\frac{1}{\eta(s)}{\rm d}s$ is
the cumulative rate of coalescence up to time $t$. 
It turns out \cite{polanski2003note} that there is an invertible linear transformation which relates $(\bbE T_{n,2}^{(\eta)},\bbE T_{n,3}^{(\eta)},\ldots,\bbE T_{n,n}^{(\eta)})$ to $\bffc^{(\eta)} \eqdef (\fc^{(\eta)}_2,\fc^{(\eta)}_3,\ldots,\fc^{(\eta)}_n)$.
Using this relation, the quantity $\sfs_{n,b}^{(\eta)}$ can be written as \cite{polanski2003new}
\begin{equation}
	\sfs_{n,b}^{(\eta)}=\frac{\left\langle \mathbf{\fc}^{(\eta)},\mathbf{W}_{n,b}\right\rangle }{\left\langle \mathbf{\fc}^{(\eta)},\mathbf{V}_{n}\right\rangle },
\label{eq:qnb_polanski}
\end{equation}
where  $\bfW_{n,b} = (W_{n,b,2},\ldots, W_{n,b,n})$  and $\bfV_n=(V_{n,2},\ldots,V_{n,n})$ are vectors of universal constants that do not depend on the population size function $\eta$, and $\ip{\cdot,\cdot}$ denotes the $l_2$-inner product.
Under model $\eta$, the quantity $\ip{\bffc^{(\eta)},\bfW_{n,b}}$ is the
total expected length of edges subtending $b$ out of $n$ individuals
sampled at time $0$, while the quantity $\ip{\bffc^{(\eta)},\bfV_n}$ is
the total expected tree length for a sample of size $n$. Both quantities
are positive for all population size functions $\eta$. For an arbitrary
population size function $\eta$, we have 
$\sum_{b=1}^{n-1} W_{n,b,m} = V_{n,m}$ for all $2 \leq m \leq n$, which 
implies
\begin{equation}
\sum_{b=1}^{n-1}\left\langle \mathbf{\fc}^{(\eta)},\mathbf{W}_{n,b}\right\rangle =\left\langle \mathbf{\fc}^{(\eta)},\mathbf{V}_{n}\right\rangle. \label{eq:Wnb-Vn equivalnce}
\end{equation}
For a constant function $\eta(t) \equiv N$, 
\begin{align}
\fc^{(\eta)}_m &= \frac{N}{a_m}, \nonumber \\ 
\ip{\bffc^{(\eta)},\bfW_{n,b}} &= \frac{2}{b} N, \label{eq:ip_constantPop}\\
\ip{\bffc^{(\eta)},\bfV_n} &= 2 N H_{n-1}, \label{eq:c.V_const}
\end{align}
where $H_{n-1} \eqdef \sum_{b=1}^{n-1} \frac{1}{b}$.

To formulate the problem, we employ the following notation. We suppose
that a sample of $n \ge 2$ randomly sampled individuals has been typed at
$s$ independent segregating sites. These data are used to form the
\emph{empirical} sample frequency spectrum, which is an $(n-1)$-tuple
$(\sfshat_{n,1},\ldots,\sfshat_{n,n-1})$, where $\sfshat_{n,b}$ denotes
the proportion of segregating sites with $b$ copies of the mutant allele
and $n-b$ copies of the ancestral allele. A \emph{frequency-based
estimator} is any statistic $\hat{\eta}$ which maps an empirical SFS to
a population size history.

\section{\label{sec:main}Main Results}
Here, we establish a minimax lower bound on the ability of \emph{any}
estimator $\hat{\eta}$ to accurately reconstruct population size
functions.

\bigskip
\subsection{A general bound on the KL divergence between two SFS distributions}
Abusing notation, we use $D(\eta\,\|\,\eta')$ to denote the
 Kullback-Leibler (KL) divergence between the probability distributions $\bfsfs_n^{(\eta)}$
and $\bfsfs_n^{(\eta')}$.  
In \sref{sec:proofs}, we prove the following general upper bound on the KL divergence between two SFS distributions:

\begin{theorem}
\label{thm:kl-bound-general}
Let $\cM$ denote a general space of population size functions and suppose $\eta,\eta'\in\cM$ satisfy  $\eta(t) = \eta'(t)$ for all $0 \leq t \leq t_c$ and 
$\max_{t>t_c}\eta(t)\le\min_{t> t_c}\eta'(t).$
Then, 
\begin{align}
	D(\eta \Vert\eta') & \le\frac{\left\langle \mathbf{\fc}^{(\eta')}-\mathbf{\fc}^{(\eta)},\mathbf{V}_n\right\rangle }{\left\langle \mathbf{\fc}^{(\eta)},\mathbf{V}_n\right\rangle }.
	 \label{eq:D_KL(1)}
\end{align}
\end{theorem}

\bigskip
\subsection{Bounds for a family of piecewise-constant models}
We now focus on a particular class of population size functions
which are easier to analyze and popular in the literature
\cite{bhaskar2014descartes,li2011inference,bhaskar2015efficient}. For a
fixed positive integer $K > 1$, let $\cM_K\subset \cM$ denote
the space of piecewise-constant size functions with exactly $K$ pieces.
A population size function $\eta$ is a member of $\cM_K$ if
and only if there exist positive real numbers $t_{1}<\cdots<t_{K-1}$ and
$N_{1},N_{2},\dots,N_{K}$ such that
\begin{equation}
	\eta(t)=\sum_{k=1}^{K}N_{k}\mathbf{1}\{t_{k-1}\le t<t_{k}\}, \label{eq:piecewise_etat}
\end{equation}
where by convention we define $t_{0}=0$ and $t_{K}=\infty$.
For such an $\eta$, define 
\begin{equation} 
S_{k}^{(\eta)}\eqdef\sum_{j=1}^{k}\frac{t_{j}-t_{j-1}}{N_{j}}.\label{eq:S_k}
\end{equation}
For $\eta \in \cM_K$, the expected waiting time $\fc^{(\eta)}_m$  defined in \eqref{eq:ETmm} is given by
\begin{equation}
	\fc^{(\eta)}_m =\frac{1}{a_{m}}\sum_{k=1}^{K}N_{k}(e^{-a_{m}S_{k-1}^{(\eta)}}-e^{-a_{m}S_{k}^{(\eta)}}).  \label{eq:ETmm_eta} 
\end{equation}
Note that since $t_{K}=\infty$,
\begin{equation}
e^{-a_{m}S_{K}^{(\eta)}}\equiv0,\quad\text{for all \ensuremath{\eta \in \cM_K}.}\label{eq:e_Sk=00003D0}
\end{equation}

To formulate our result, we let $I,J$ denote positive
integers that satisfy $\II+J=K$, and introduce a subfamily
$\subfamilyIJ\subset\cM_K$ of piecewise-constant functions
defined as follows. See \fref{fig:U_K} for illustration. We assume that
all change points $t_1 < \cdots < t_{I+J-1}$ are fixed and that the
sizes $N_1,\ldots, N_{\II}$ of the first $I$ epochs are also fixed, with
$N_{\II}$ being the smallest size. So, all functions in $\subfamilyIJ$
are identical to each other for the first $I$ epochs, and there is a
population bottleneck in the last epoch. Then, for $t \geq t_\II$, every
function $\eta\in\subfamilyIJ$ undergoes jumps according to the following
rules:
\begin{enumerate}
	\item For the interval $t_I \leq t < t_{I+1}$, $\eta(t)$ takes a constant value of either $h$ or $h+\delta$, where $h > N_{\II}$   and $\delta > 0$.
	\item At later change points $\{t_{I+1},\ldots,t_{I+J-1}\}$, $\eta$ either stays the same or jumps upward by $\delta$.
\end{enumerate}
Hence, $\subfamilyIJ$ consists of $2^J$ distinct piecewise-constant
functions that are non-decreasing functions of $t$ for $t \geq
t_I$. Note that $\min_{t} \eta(t) = N_{\II}$ for all $\eta \in
\subfamilyIJ$. For ease of notation, we use $\vep\eqdef N_{\II}$ to
denote the bottleneck size and $\tb\eqdef t_{\II} - t_{\II-1}$ to
denote the bottleneck duration. To facilitate analysis later, we
fix $t_{\II+j}-t_{\II+j-1}$ to some positive constant $\ta$ for
all $j=1,\ldots, J-1$.

\begin{figure}[t]
\centering 
\includegraphics[width=\textwidth]{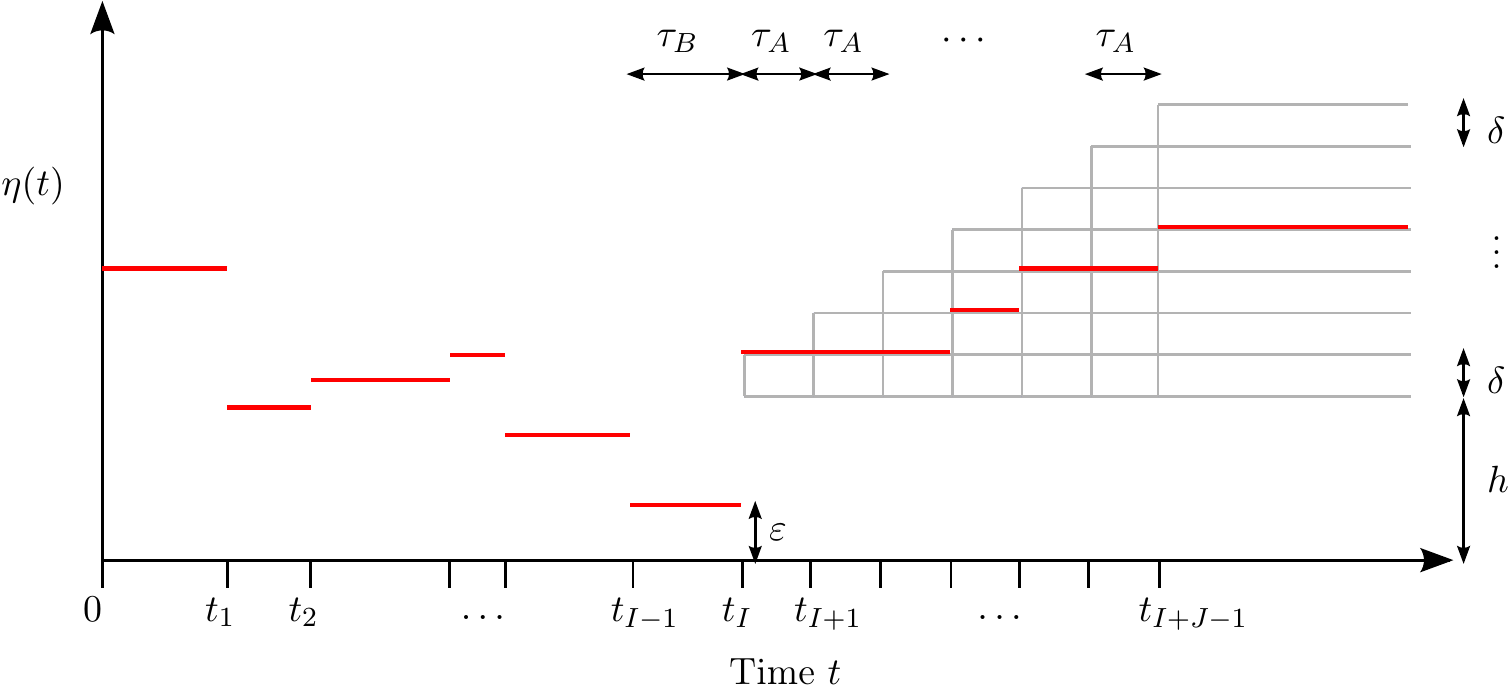}
\caption{A family $\subfamilyIJ$ of piecewise-constant population size models with $K = I+J$ epochs.} 
\label{fig:U_K}
\end{figure}

For any two models in $\subfamilyIJ$, we obtain the following bound on the difference of their waiting times to the first coalescence:
\begin{lemma}
\label{lem:e_bound}
For all $\eta,\eta'\in\subfamilyIJ$,
\begin{equation}
\big|\fc_{m}^{(\eta)}-\fc_{m}^{(\eta')}\big| \leq J \frac{\delta}{a_m} e^{-a_m \tb/\vep}. \label{eq:e1-e2 bound}
\end{equation}
\end{lemma}

\noindent
Then, together with \thmref{thm:kl-bound-general}, this lemma can be used to show
\begin{theorem}
\label{thm:kl-bound}Let $\eta, \eta' \in  \subfamilyIJ$ that satisfy $\max_{t\geq \tII}\eta(t)\le\min_{t\geq \tII}\eta'(t).$  Then,
\begin{equation}
D(\eta\,\|\,\eta')\leq J \frac{\delta}{\vep} e^{-\tb/\vep}.
\label{eq:kl-bound}
\end{equation}
 \end{theorem}
\noindent
Our proofs of the above results are deferred to \sref{sec:proofs}. It is
interesting that the above bound does not depend on the number $n$ of
sampled individuals.

\bigskip
\subsection{Minimax lower bounds}
Before using the above results to obtain a minimax lower bound, we
first note a subtle fact. Given any population size function $\eta$,
consider a function $\zeta$ which satisfies $\zeta(t) = \kappa\cdot
\eta(t/\kappa)$ for all $t \in [0,\infty)$, where $\kappa$ is some
positive constant. Such functions are equivalent, as it turns out that
$\sfs_{n,b}^{(\zeta)} = \sfs_{n,b}^{(\eta)}$ for all $n \geq 2$ and $1
\leq b \leq n-1$. To mod out by this equivalence, we assume that every
$\eta\in\cM$ satisfies $\eta(0)=\Nfix$, where $\Nfix$ is some fixed
positive constant.

Let $\norm{\cdot}{*}$ denote a generic norm (specific examples will
be given later) and let $\mathbb{E}_{\eta}(\cdot)$
denote expectation with respect to the SFS distribution
$\bfsfs^{(\eta)}_n=(\sfs^{(\eta)}_{n,1},\ldots,\sfs^{(\eta)}_{n,n-1})$
induced by population size function $\eta$. Then, note that
\[
\inf_{\hat\eta} \sup_{\eta\in\cM}\mathbb{E}_{\eta} \norm{\hat{\eta}-\eta}{*}  \geq 
\inf_{\hat\eta} \sup_{\eta\in\cM_K}\mathbb{E}_{\eta} \norm{\hat{\eta}-\eta}{*} \geq
\inf_{\hat\eta} \sup_{\eta\in\subfamilyIJ}\mathbb{E}_{\eta} \norm{\hat{\eta}-\eta}{*}.
\]
In what follows, we will put a lower bound on the last
quantity. We first fix a sensible distance metric on
$\cM$. An intuitive way to measure distance between two
population size functions is their $L_{1}$ distance,
$\|\eta_{a}-\eta_{b}\|_{1}=\int_{0}^{\infty}|\eta_{a}(t)-\eta_{b}(t)| \mathrm{d}t$, 
but this is unreasonably stringent in that 
$\left\Vert \eta_{a}-\eta_{b}\right\Vert _{1}=\infty$ 
if $\eta_{a}$ and $\eta_{b}$ do not agree infinitely far back into the
past. Instead we will focus on the following truncated $L_{1}$ distance:
$\|\eta_{a}-\eta_{b}\|_{1,T}\eqdef\int_{0}^{T}|\eta_{a}(t)-\eta_{b}(t)|\mathrm{d}t$, 
which measures the discrepancy between $\eta_{a}$ and
$\eta_{b}$ back to some fixed time $T$ in the past.

Henceforth, let $\hat{\eta}$ be any
estimator of the population size function which operates on a sample of
$s$ independent segregating sites obtained from a sample of $n$ randomly sampled
individuals. 
In \sref{sec:proofs}, we prove the following main results of our paper:

\begin{theorem}
\label{thm:nonasymptotic}
Consider the subfamily $\subfamilyIJ$ of models described above, and
suppose $J > 8$ and $T \geq t_{\II+J-1}+\ta$. 
Then, 
\begin{equation}
\inf_{\hat\eta} \sup_{\eta\in\subfamilyIJ}\mathbb{E}_{\eta} \norm{\hat{\eta}-\eta}{1,T} \geq 
C \ta \frac{(J - 8)^2}{J} \frac{\vep}{s}e^{\tb / \vep},
\end{equation}
where $C$ is a positive constant.
\end{theorem}

The above theorem applies to all models in $\subfamilyIJ$.  We now consider the subset $\subfamilyIJ^M = \{\eta \in \subfamilyIJ: \Vert \eta \Vert_\infty < M\}$, which is the set of all models in $\subfamilyIJ$ that are bounded by some constant $M$.  For this family of bounded population size functions, a sharper asymptotic lower bound can be obtained as follows.

\begin{theorem}\label{thm:minimax_UK}
    Suppose $J > 8$ and $T \geq t_{\II+J-1} + \ta$.
    Then,
\begin{equation}\label{eq:main_result}
\inf_{\hat\eta} \sup_{\eta\in\subfamilyIJ^M}\mathbb{E}_{\eta} \norm{\hat{\eta}-\eta}{1,T} \geq C' \frac{(J-8)^2}{J} \frac{\tb\ta}{\log s},
\end{equation}
where $C'$ is a positive constant.
\end{theorem}

By specializing $\subfamilyIJ^M$, a simplified version of \thmref{thm:minimax_UK}
can be obtained:

\begin{corollary}
\label{cor:corollary}
Suppose $T \geq t_{\II+J-1} + \ta$ and let $\subfamily_{I,\star}^M = \bigcup_{J\ge 1}\subfamilyIJ^M$. Then,
\begin{equation}
    \inf_{\hat\eta} \sup_{\eta\in\subfamily_{I,\star}^M}\mathbb{E}_{\eta} \norm{\hat{\eta}-\eta}{1,T} 
        \geq C'' (T - t_I) \frac{\tb}{\log s},
\end{equation}
where $C''$ is a positive constant.
\end{corollary}

Note that the above lower bounds do not depend on the dimension of the SFS  (which is equal to $n-1$).
Hence, for a fixed number $s$ of segregating sites considered, using more individuals does not help to reduce the error bounds.

\section{\label{sec:discuss}Discussion}
In this paper, we have theoretically characterized fundamental limits on the accuracy of demographic inference from data. 
The paper that most closely relates to the
present work is by Kim \emph{et al.}~\cite{Kim201526}, who obtain lower
bounds on the amount of \emph{exact} coalescence time data necessary to
distinguish between size histories in a hypothesis testing framework.
Since coalescence times are never observed and must be estimated
from data, these bounds place a limit on the accuracy with which a
population size function can be inferred. The authors also describe an
estimator which uses coalescence times (again observed without noise)
to accurately recover the underlying population size function with
high probability, at a rate that roughly matches the lower bound.
The results presented here are complementary to those of Kim \emph{et al.}~in 
the following sense: access to coalescence times is equivalent
to knowing the site frequency spectrum, so that their results apply
to SFS-based estimators as well as other types of procedures which
also incorporate e.g.~linkage information. By focusing specifically on
SFS-based estimators, we are able to more precisely quantify the effects
of sample size and underlying demography on the ability to accurately
reconstruct a population size function. Ultimately, we derive minimax
lower bounds on the error of estimating a size history using the SFS.

Another line of work centers around the \emph{identifiability} of
the parameter $\eta(t)$ using the SFS. Roughly speaking, a family of
statistical models $\{P_\theta\}_{\theta \in \Theta}$ defined over a
parameter space $\Theta$ is identifiable if for any $\theta_1,\theta_2
\in \Theta$ with $\theta_1 \neq \theta_2$, the sampling distributions
induced by $P_{\theta_1}$ and $P_{\theta_2}$ are different. In our
context this simply says that, for all $n$, $\bfsfs_n^{(\eta_1)} \neq
\bfsfs_n^{(\eta_2)}$ unless $\eta_1 = \eta_2$ almost everywhere.
Standard desiderata for statistical estimators (e.g., consistency or
unbiasedness) are impossible without identifiability, so it is the
weakest possible regularity condition one can impose on a useful family
of models.

Perhaps surprisingly, it turns out that in general a population size function is not identifiable
from the SFS \cite{myers2008can}. Indeed, for any given $\eta(t)$, it has been shown
that an infinite number of smooth functions $F(t)$ exist such that
$\bfsfs_n^{(\eta)} = \bfsfs_n^{(\eta + F)}$.
Moreover, explicit
examples can be constructed which demonstrate this phenomenon
\cite{myers2008can}. On the other hand, these counter-examples consist
of functions which exhibit an unbounded frequency of oscillatory behavior near the present time, which is
perhaps unrealistic when modeling naturally occurring populations. More
recently, it has been shown \cite{bhaskar2014descartes} that identifiability holds for many classes
of population size functions employed by practitioners (including
piecewise-constant, -exponential and -generalized-exponential).
Furthermore, the number $n$ of sampled individuals sufficient for identifiability
can be explicitly given and is a function of the complexity of the
underlying class of models being studied \cite{bhaskar2014descartes}.

Identifiability asserts that, given an \emph{infinite} amount of data
(specifically, taking the number of segregating sites $s\to\infty$), the
model parameter $\eta(t)$ can be uniquely recovered. In practice, $s$ is
finite and only a perturbed version of the expected frequency spectrum,
say $\hat{\bfsfs}_n^{(\eta)}$, is observed. From a practical standpoint,
it is important to understand how these perturbations ultimately affect
the parameter estimate $\hat{\eta}(t)$. It is this question that forms
the starting point for the present work.

We have shown that the minimax error rate for estimating the
piecewise-constant demography of a single population is at least $O(1
/ \log s)$, where $s$ is the number of independent segregating sites
analyzed. In contrast, the minimax error for many classical estimation
problems in statistics (for example, non-parametric regression or density
estimation) decays inverse-polynomially in the sample size \cite{tsybakov}. 
Compared with these problems, exponentially more samples would be
required to estimate a population size history function to within a
similar magnitude of error.

A single population evolving under a piecewise-constant demography
is a special case of many richer classes of demographic models.
For example, it is a (limiting) member of the family of
exponential growth models, seen by taking each exponential growth
parameter to zero. In the multi-species coalescent setting
\cite{chen2012joint,excoffier2013robust}, multiple population size
histories must be estimated, and the error of that estimate must
necessarily be lower bounded by that of estimating a single such
history. Thus, our result can be expected to apply to a broader class of
models than the one we have studied here.

As detailed in \sref{sec:proofs}, the result in \thmref{thm:minimax_UK}
follows from setting $\vep = \tb/\log s$ and $\delta \propto \frac{\vep}{s} \exp(\tau_B/\vep)$ in
the subfamily $\subfamilyIJ^M$. The size $\tb/\log s$ is in coalescent
units. In terms of the number of individuals, it is proportional to
$\gb/\log s$, where $\gb$ is the number of generations corresponding to
duration $\tb$ in the coalescent limit. Intuitively, as the severity
of the bottleneck increases, the population is increasingly likely to
find its most recent common ancestor (MRCA) during that time; further
back in time than the MRCA, no information is conveyed concerning the
demographic events experienced by the population.

One might object to considering models with a bottleneck size that
scales inversely with the number $s$ of segregating sites in the
data, and it is indeed possible that a better convergence rate may be
achievable for populations which are known not to contain a bottleneck.
On the other hand, we note that $1/\log s$ decreases sufficiently
slowly with $s$ that our result can be expected to apply to many
real-world examples. For example, for $s\approx 10^{8}$, which is
a conservative upper bound for most organisms, $\gb/\log s \approx
0.054 \gb$. This implies that for populations which have experienced
roughly an order-of-magnitude increase in effective population size
during their history, accurate estimation of demographic events which
occurred before this expansion is difficult using SFS-based methods.
Additionally, an interesting aspect of our work is that our minimax lower
bounds do not depend on the number $n$ of sampled
individuals; increasing $n$ is not enough to overcome
the information barrier imposed by the presence of a bottleneck. This is
intuitively plausible since, as $n$ increases, the $(n+1)$th sampled
lineage becomes more likely to coalesce early on.

An interesting question which we have not attempted to analyze is
whether the $O(1 / \log s)$ rate is optimal, i.e.~whether there
exists some estimator $\hat{\eta}(t)$ which achieves the minimax
lower bound established here. In practice, from equations \eqref{eq:qnb_polanski},
\eqref{eq:S_k}, and \eqref{eq:ETmm_eta}, it can be seen that naively
maximizing the likelihood of the observed SFS with respect to $\eta(t)$
requires solving a non-convex optimization problem, so that convergence
to the global maximum is not even guaranteed. Computational issues aside,
finding such an estimator remains an open theoretical challenge.

In closing, we stress that our result is specific to SFS-based
estimators, which analyze only independent sites.
The main allure of these estimators is their mathematical
tractability, rather than their realism. In fact, a rich source
of additional information exists in the correlation structure
found among linked sites in the genome. Methods which seek to
exploit this structure by modeling the action of recombination
pose greater mathematical and computational
difficulties, but there has been recent progress in this area
\cite{li2011inference,paul2011accurate,sheehan2013estimating,steinruecken2013sequentially,rasmussen2014genome,schiffels2014inferring}.
Our result serves to underscore the importance of pursuing ever more
realistic models of genomic evolution, challenging though they may be.

\section{Proofs}
\label{sec:proofs}

\begin{proof}[Proof of \thmref{thm:kl-bound-general}]
	To simplify the notation, we write $\bffc=\bffc^{(\eta)}$ and $\bffc'=\bffc^{(\eta')}$.	 Then, using \eqref{eq:qnb_polanski} we can write
\begin{equation*}
	D(\eta\Vert\eta') =\sum_{b=1}^{n-1}\sfs_{n,b}^{(\eta)}\log\frac{\sfs_{n,b}^{(\eta)}}{\sfs_{n,b}^{(\eta')}}
	 =\sum_{b=1}^{n-1}\sfs_{n,b}^{(\eta)}\left[\log\left(\frac{\left\langle \bffc,\mathbf{W}_{n,b}\right\rangle }{\left\langle \bffc',\mathbf{W}_{n,b}\right\rangle }\right)+\log\left(\frac{\left\langle \bffc',\mathbf{V}_n \right\rangle }{\left\langle \bffc,\mathbf{V}_n\right\rangle }\right)\right].
\end{equation*}
The assumption $\min_{t > t_c}\eta'(t)\ge\max_{t> t_c}\eta(t)$ implies
that, for all times $t,t' > t_c$, the instantaneous rate of coalescence
at time $t$ in model $\eta$ is $\geq$ the instantaneous rate of
coalescence at time $t'$ in model $\eta'$. Hence, this assumption
together with $\eta(t) = \eta'(t)$ for all $0\leq t \leq t_c$ implies
$\ip{\bffc-\bffc',\bfW_{n,b}} \leq 0$ for all $1 \leq b \leq n-1$; equivalently,
    $\log\left(\frac{\left\langle \bffc,\mathbf{W}_{n,b}\right\rangle }{\left\langle \bffc',\mathbf{W}_{n,b}\right\rangle }\right) < 0$.
    Additionally, $\frac{\left\langle \bffc'-\bffc,\mathbf{V}_n\right\rangle }{\left\langle \bffc,\mathbf{V}_n\right\rangle } > -1$ and $\log(1+x) \leq x$ for all $x\geq -1$. Combining these facts, we obtain
\begin{equation*}
    D(\eta\Vert\eta') \le 
    \sum_{b=1}^{n-1}\sfs_{n,b}^{(\eta)}\log\left(\frac{\left\langle \bffc',\mathbf{V}_n \right\rangle }{\left\langle \bffc,\mathbf{V}_n\right\rangle }\right)
    \le \sum_{b=1}^{n-1}\sfs_{n,b}^{(\eta)}\frac{\left\langle \bffc'-\bffc,\mathbf{V}_n\right\rangle }{\left\langle \bffc,\mathbf{V}_n\right\rangle } 
    = \frac{\left\langle \bffc'-\bffc,\mathbf{V}_n\right\rangle }{\left\langle \bffc,\mathbf{V}_n\right\rangle },
\end{equation*}
where we have used $\sum_{b=1}^{n-1}\sfs_{n,b}^{(\eta)}=1$ in the final equality.
\end{proof}

\begin{proof}[Proof of \lemref{lem:e_bound}]
We distinguish two particular models $\etalower,\etaupper\in\subfamilyIJ$ which are the \emph{lower} and the \emph{upper} envelopes of $\subfamilyIJ$.   
The function $\etalower$ stays constant at $h$ for all $t \geq \tII$, while $\etaupper$ jumps upward by $\delta$ at every change point $t_{\II},\ldots,t_{\II+J-1}$.  Hence, $\etalower\le\eta\le\etaupper$ pointwise for all $\eta\in\subfamilyIJ$.
The two enveloping functions will form the basis of subsequent analysis.  

Fix $\eta,\eta'\in \subfamilyIJ$ and note that, by the definition of
$\subfamilyIJ$, one of these functions must pointwise-dominate the other.
So, assume without loss of generality that $\eta(t)\le\eta'(t)$ for all
$t$. Then for all $t$,
\[
	\etalower(t)\le\eta(t)\le\eta'(t)\le\etaupper(t),
\] 
which implies
\[
	\fc_{m}^{(\etalower)}\le \fc_{m}^{(\eta)}\le \fc_{m}^{(\eta')}\le \fc_{m}^{(\etaupper)},
\]
for all $m=2,\ldots,n$.  Using these inequalities, we conclude
\[
	\fc_{m}^{(\eta')}-\fc_{m}^{(\eta)}\le \fc_{m}^{(\etaupper)}-\fc_{m}^{(\etalower)},
\]
so it suffices to demonstrate \eqref{eq:e1-e2 bound} for $\fc_{m}^{(\etaupper)}-\fc_{m}^{(\etalower)}$.
Now, by equation (\ref{eq:ETmm_eta}) and the definition of $\etalower$,
\begin{align*}
	a_{m}\fc_{m}^{(\etalower)} & = \sum_{i=1}^{\II} N_i \left[e^{-a_m \Sk{\etalower}{i-1}} - e^{-a_m \Sk{\etalower}{i}}\right] + 
	\sum_{j=1}^{J} h \left[e^{-a_m \Sk{\etalower}{I+j-1}} - e^{-a_m \Sk{\etalower}{I+j}}\right]\\
	 & = \sum_{i=1}^{\II} N_i \left[e^{-a_m \Sk{\etalower}{i-1}} - e^{-a_m \Sk{\etalower}{i}}\right] + h e^{-a_m \Sk{\etalower}{\II}},
\end{align*}
where we have used equation (\ref{eq:e_Sk=00003D0}). Similarly,
\begin{align*}
	a_{m}\fc_{m}^{(\etaupper)}  & = \sum_{i=1}^{\II} N_i \left[e^{-a_m \Sk{\etaupper}{i-1}} - e^{-a_m \Sk{\etaupper}{i}}\right]\\
	& \hspace{1cm} + 
	\sum_{j=1}^{J} (h+j\delta) \left[e^{-a_m \Sk{\etaupper}{I+j-1}} - e^{-a_m \Sk{\etaupper}{I+j}}\right]\\
	 & = \sum_{i=1}^{\II} N_i \left[e^{-a_m \Sk{\etaupper}{i-1}} - e^{-a_m \Sk{\etaupper}{i}}\right] 
	 + h e^{-a_m \Sk{\etaupper}{\II}}\\
	&\hspace{1cm} +
	\sum_{j=1}^{J} j\delta \left[e^{-a_m \Sk{\etaupper}{I+j-1}} - e^{-a_m \Sk{\etaupper}{I+j}}\right].
\end{align*}
Now, using the fact that $\etalower$ and $\etaupper$ agree on the first $I$ epochs, we obtain
\begin{align*}
	a_m[ \fc_{m}^{(\etaupper)}-\fc_{m}^{(\etalower)}] & = \sum_{j=1}^{J} j\delta \left[e^{-a_m \Sk{\etaupper}{I+j-1}} - e^{-a_m \Sk{\etaupper}{I+j}}\right] \\
    & = \delta \sum_{j=1}^{J} e^{-a_m \Sk{\etaupper}{I+j-1}}\\
    & \le J \delta e^{-a_m \tb/\vep},
\end{align*}
where the second line follows from telescoping and the fact $\Sk{\etaupper}{I+J}=\infty$, while the last line follows from the fact that  $\frac{\tb}{\vep} \leq \Sk{\etaupper}{I+j-1}$ for all $j=1,\ldots, J$.
\end{proof}

\begin{proof}[Proof of \thmref{thm:kl-bound}]
For ease of notation, define $\bffc=\bffc^{(\eta)}$ and $\bffc'=\bffc^{(\eta')}$.  By Lemma \ref{lem:e_bound},
\begin{align*}
	\ip{\bffc'-\bffc,\mathbf{V}_n} = 
	\sum_{m=2}^{n}(\fc'_{m}-\fc_{m})V_{n,m}
	& \leq J \delta \sum_{m=2}^{n} \frac{V_{n,m}}{a_m} e^{-a_m \tb/\vep}\\
	& \leq J \delta e^{-\tb/\vep} \sum_{m=2}^{n} \frac{V_{n,m}}{a_m},
\end{align*}
where the second inequality follows from $e^{-a_m \tb/\vep} \leq e^{-\tb/\vep}$ for all $m=2,\ldots,n$.
Now, noting that $\sum_{m=2}^{n} \frac{V_{n,m}}{a_m}$ corresponds to the total tree length for the constant population size function $\eta\equiv 1$ and using \eqref{eq:c.V_const}, we obtain
\begin{equation}
	\ip{\bffc'-\bffc,\mathbf{V}_n} \leq J \delta e^{-\tb/\vep} 2 H_{n-1}.
	\label{eq:numerator}
\end{equation}
To finish the proof, recall that $\ip{\bffc,\mathbf{V}_n}$
is the total expected branch length of the coalescent tree under model
$\eta$.  Since  $\min_{t}\eta(t) = \vep,$ we have that $\left\langle \bffc,\bfV_n\right\rangle $
is at least as large as the corresponding quantity under a model with constant
population size $\vep$.  By \eqref{eq:c.V_const}, the total expected tree length under the latter model equals
$2 \vep H_{n-1}$. Thus, $\ip{\bffc,\bfV_n} \geq 2\vep H_{n-1}$ and combining this result with \eqref{eq:numerator} gives
\[
	\frac{\left\langle \bffc'-\bffc,\mathbf{V}_n\right\rangle}{\ip{\bffc,\bfV_n}} 
	\leq J \frac{\delta}{\vep} e^{-\tb/\vep}.
\]
Finally, \eqref{eq:kl-bound} follows from this inequality and \thmref{thm:kl-bound-general}.
\end{proof}

\begin{proof}[Proof of \thmref{thm:nonasymptotic}]
Our proof uses a generalized form of Fano's inequality \cite{yu1997assouad}.
Adapted to our setting and notation, the method reads as follows.
\begin{theorem}[Fano's method]
\label{thm:fano}
Consider a space $\cM$ of population size models. Let $r\ge2$ be
an integer, and let $S^{r}_n=\{ \eta_1,\eta_2,\dots,\eta_{r}\}
\subset \cM$ contain $r$ population size functions such that for
all $a\neq b$, $\Vert\eta_{a}-\eta_{b}\Vert_{*}\ge \alpha_{r}$ and
$D(\sfs_n^{(\eta_{a})}\Vert \sfs_n^{(\eta_{b})})\le \beta_{r}$. Let
$\hat{\eta}^{(n,s)} = \hat{\eta}^{(n,s)}(X_1,\ldots,X_s)$ be an
estimator of $\eta$ based on the SFS data $X_1,\ldots,X_s$ sampled
independently from $\sfs^{(\eta)}_n$; i.e., $X_1,\ldots,X_s$ are SFS
data for $n$ individuals at $s$ independent segregating sites. Then,
\begin{equation}
    \label{eq:technical_bound}
\inf_{\hat{\eta}} \sup_{\eta \in \cM}\mathbb{E}_{\eta}\Vert\hat{\eta}^{(n,s)}-\eta\Vert_{*} \ge   \frac{\alpha_{r}}{2}\Big(1-\frac{s\cdot \beta_{r}+\log2}{\log r}\Big).
\end{equation}
\end{theorem}

This theorem places a lower bound on the minimax rate of convergence of
a population size history estimator based on the SFS.

For $\eta\in\subfamilyIJ$, let $w_j$ denote the variable $\in\{0,1\}$
indicating whether $\eta$ jumps by $\delta$ at change point
$t_{\II+j}$. Let $\cY = \{ \bfw=(w_0,\ldots,w_{J-1}) \mid w_i
\in \{0,1\} \}$, where $J\geq 8$. By the Varshamov-Gilbert
lemma (see \cite[Lemma 4.7]{massart2007concentration}),
there exist $\cX=\{\bfw^0,\ldots,\bfw^M\} \subset \cY$ such
(i) $\bfw^0=(0,\ldots,0)$, (ii) $M \geq 2^{J/8}$, and (iii)
$H(\bfw^i,\bfw^j) \geq J/8$, where $H(\cdot,\cdot)$ denotes the Hamming
distance.

Let $\subfamilyIJ^\cX$ denote the subset of $2^{J/8}+1$ functions
in $\subfamilyIJ$ with the indicator variable for $\delta$-jumps at
$t_{\II},\ldots,t_{\II+J-1}$ given by $\bfw\in\cX$. Then, for any two
$\eta_{a}\neq\eta_{b}\in\subfamilyIJ^\cX$, we have
\begin{equation}
\norm{\eta_{a}-\eta_{b}}{1,T}\ge \frac{J}{8} \cdot \ta\cdot \delta .
\label{eq:norm_lb}
\end{equation}

\noindent
Using \thmref{thm:fano} via \eqref{eq:norm_lb} and \thmref{thm:kl-bound}, we obtain
\begin{align}
\inf_{\hat{\eta}} \sup_{\eta \in \subfamilyIJ}\mathbb{E}_{\eta}\Vert\hat{\eta}^{(n,s)}-\eta\Vert_{1,T}  
&\ge \frac{J\cdot\ta\cdot \delta}{16}  \left[1 - \frac{s J \frac{\delta}{\vep} e^{-\tb/\vep} + \log 2}{\log(2^{J/8}+1)} \right] \nonumber \\
&\ge \frac{J\cdot\ta\cdot \delta}{16}  \left[1 - \frac{s J \frac{\delta}{\vep} e^{-\tb/\vep} + \log 2}{ \frac{J}{8}\log 2} \right].\label{eq:lb_with_delta}
\end{align}
We now optimize the bound with respect to $\delta$. A straightforward
calculation shows that the maximum is attained at
\begin{equation}\label{eq:delta_star}
    \delta^{*}=  \frac{(J-8) \log 2}{16J} \left( \frac{\vep}{s}\right) e^{\tb / \vep}
\end{equation}
and setting $\delta = \delta^{*}$ in \eqref{eq:lb_with_delta} yields the result.
\end{proof}

\begin{proof}[Proof of \thmref{thm:minimax_UK}]
    The result is obtained by scaling $\vep$ with the number of
    segregating sites $s$. Denote this scaling by $\vep(s)$; we will
    determine $\vep(s)$ which produces the largest possible lower bound.
    Starting from Equation \eqref{eq:delta_star} in the proof of
    \thmref{thm:nonasymptotic}, note that $\delta^*$ scales as
    $\frac{\vep}{s}e^{\tb / \vep} =: f(\vep)$. In order to satisfy the constraint
    that $\Vert \eta \Vert_\infty < M$ for all $\eta \in \subfamilyIJ^M$ and $s$,
    the condition
    \begin{equation}
        \label{eq:bounded_growth}
        \limsup_{s\to\infty} \max\Big\{\frac{\vep(s)}{s}e^{\tb / \vep(s)}, \vep(s)\Big\} < \infty
    \end{equation}
    must therefore hold. This implies that $\vep(s)s^p \to \infty$ as $s\to\infty$
    for all $p > 0$. Suppose that $q \eqdef \liminf_{s\to\infty} \frac{\vep(s) \log s}{\tau_B}
    < 1$; note that $\vep(s)>0$ implies $q > 0$. Then there exists
    a diverging sequence $s_1,s_2,\dots \to \infty$ with $\log(s_i) <
    \frac{1 + q}{2}\frac{\tau_B}{\vep(s_i)}$ for all $i$, whence
    \[
        \limsup_{s\to\infty} \frac{\vep(s)}{s}e^{\tau_B / \vep(s)} \ge 
        \limsup_{i\to\infty} \frac{\vep(s_i)}{s_i}e^{\frac{2}{1+q}\log(s_i)} =
        \limsup_{i\to\infty} \vep(s_i)s_i^{\frac{1-q}{1+q}} = \infty.
    \]
    From this it follows that $\vep(s) \ge \tb / \log s$ for
    sufficiently large $s$. Now, on the interval $(0, \infty)$, the
    function $f(\vep)$ is convex with a unique minimum at $\vep = \tb$.
    Let $\vep'$ be a point where $f(\vep') > f(\tb / \log s) = \tb /
    \log s$. Then $\vep' \notin [\frac{\tb}{\log s}, \tb]$. If $\vep' > \tb$,
    then $f(\vep') < \frac{\vep'}{s}e^1$.  Since $\frac{\tb}{\log s} < f(\vep')$, we then conclude $\vep' >
    \frac{s \tb}{e^1 \log s}$, which is not bounded as $s\to\infty$.

    In summary, we see that the largest possible lower bound which obeys 
    \eqref{eq:bounded_growth} must have $f(\vep)$ asymptotically $\le \tb / \log s$,
    and that this bound is achieved by setting $\vep(s)=\tb / \log s$.
    Plugging this in to equation \eqref{eq:technical_bound} yields the claim.
\end{proof}

\begin{proof}[Proof of Corollary \ref{cor:corollary}]
	For $c \in (0, 1)$, choose $J$ large enough so that $(J - 8)/ J > c$, and fix $\ta$ so that $T = t_I + J \ta$. Then $(J - 8)\ta \ge cJ \ta = c(T - t_I)$. 
	Substituting the above inequalities into equation \eqref{eq:main_result} and letting $C'' = C'c^2$ yields the desired result.
\end{proof}

\section*{Acknowledgments}
We thank Anand Bhaskar for helpful comments on a draft of this paper and for suggesting Corollary~\ref{cor:corollary} to simplify the presentation of the main result.  We also thank Jack Kamm and Jeff Spence for useful feedback.  

\clearpage

\bibliographystyle{imsart-number} 
\bibliography{yss-group,minimax}

\providecommand{\sortkey}[1]{}
\begin{thebibliography}{0}

\bibitem{gpc2010map}
\begin{barticle}[author]
\bauthor{\bsnm{{1000 Genomes Project Consortium}},~}
(\byear{2010}).
\btitle{A map of human genome variation from population-scale sequencing}.
\bjournal{Nature}
\bvolume{467}
\bpages{1061-1073}.
\end{barticle}
\endbibitem

\bibitem{bhaskar2014descartes}
\begin{barticle}[author]
\bauthor{\bsnm{Bhaskar},~\bfnm{A}\binits{A.}} \AND
  \bauthor{\bsnm{Song},~\bfnm{Yun~S}\binits{Y.~S.}}
(\byear{2014}).
\btitle{Descartes' rule of signs and the identifiability of population
  demographic models from genomic variation data}.
\bjournal{Annals of Statistics}
\bvolume{42}
\bpages{2469--2493}.
\end{barticle}
\endbibitem

\bibitem{bhaskar2015efficient}
\begin{barticle}[author]
\bauthor{\bsnm{Bhaskar},~\bfnm{A.}\binits{A.}},
  \bauthor{\bsnm{Wang},~\bfnm{Y.~X.~Rachel}\binits{Y.~X.~R.}} \AND
  \bauthor{\bsnm{Song},~\bfnm{Y.~S.}\binits{Y.~S.}}
(\byear{2015}).
\btitle{Efficient inference of population size histories and locus-specific
  mutation rates from large-sample genomic variation data}.
\bjournal{Genome Research}
\bvolume{25}
\bpages{268-279}.
\end{barticle}
\endbibitem

\bibitem{chen2012joint}
\begin{barticle}[author]
\bauthor{\bsnm{Chen},~\bfnm{Hua}\binits{H.}}
(\byear{2012}).
\btitle{The joint allele frequency spectrum of multiple populations: A
  coalescent theory approach}.
\bjournal{Theoretical Population Biology}
\bvolume{81}
\bpages{179--195}.
\end{barticle}
\endbibitem

\bibitem{coventry2010deep}
\begin{barticle}[author]
\bauthor{\bsnm{Coventry},~\bfnm{Alex}\binits{A.}},
  \bauthor{\bsnm{Bull-Otterson},~\bfnm{Lara~M}\binits{L.~M.}},
  \bauthor{\bsnm{Liu},~\bfnm{Xiaoming}\binits{X.}},
  \bauthor{\bsnm{Clark},~\bfnm{Andrew~G}\binits{A.~G.}},
  \bauthor{\bsnm{Maxwell},~\bfnm{Taylor~J}\binits{T.~J.}},
  \bauthor{\bsnm{Crosby},~\bfnm{Jacy}\binits{J.}},
  \bauthor{\bsnm{Hixson},~\bfnm{James~E}\binits{J.~E.}},
  \bauthor{\bsnm{Rea},~\bfnm{Thomas~J}\binits{T.~J.}},
  \bauthor{\bsnm{Muzny},~\bfnm{Donna~M}\binits{D.~M.}},
  \bauthor{\bsnm{Lewis},~\bfnm{Lora~R}\binits{L.~R.}} \AND \betal{et~al.}
(\byear{2010}).
\btitle{Deep resequencing reveals excess rare recent variants consistent with
  explosive population growth}.
\bjournal{Nature Communications}
\bvolume{1}
\bpages{131}.
\end{barticle}
\endbibitem

\bibitem{excoffier2013robust}
\begin{barticle}[author]
\bauthor{\bsnm{Excoffier},~\bfnm{Laurent}\binits{L.}},
  \bauthor{\bsnm{Dupanloup},~\bfnm{Isabelle}\binits{I.}},
  \bauthor{\bsnm{Huerta-S{\'a}nchez},~\bfnm{Emilia}\binits{E.}},
  \bauthor{\bsnm{Sousa},~\bfnm{Vitor~C}\binits{V.~C.}} \AND
  \bauthor{\bsnm{Foll},~\bfnm{Matthieu}\binits{M.}}
(\byear{2013}).
\btitle{Robust Demographic Inference from Genomic and {SNP} Data}.
\bjournal{PLoS Genetics}
\bvolume{9}
\bpages{e1003905}.
\end{barticle}
\endbibitem

\bibitem{gazave2014neutral}
\begin{barticle}[author]
\bauthor{\bsnm{Gazave},~\bfnm{Elodie}\binits{E.}},
  \bauthor{\bsnm{Ma},~\bfnm{Li}\binits{L.}},
  \bauthor{\bsnm{Chang},~\bfnm{Diana}\binits{D.}},
  \bauthor{\bsnm{Coventry},~\bfnm{Alex}\binits{A.}},
  \bauthor{\bsnm{Gao},~\bfnm{Feng}\binits{F.}},
  \bauthor{\bsnm{Muzny},~\bfnm{Donna}\binits{D.}},
  \bauthor{\bsnm{Boerwinkle},~\bfnm{Eric}\binits{E.}},
  \bauthor{\bsnm{Gibbs},~\bfnm{Richard~A}\binits{R.~A.}},
  \bauthor{\bsnm{Sing},~\bfnm{Charles~F}\binits{C.~F.}},
  \bauthor{\bsnm{Clark},~\bfnm{Andrew~G}\binits{A.~G.}} \AND \betal{et~al.}
(\byear{2014}).
\btitle{Neutral genomic regions refine models of recent rapid human population
  growth}.
\bjournal{Proceedings of the National Academy of Sciences}
\bvolume{111}
\bpages{757--762}.
\end{barticle}
\endbibitem

\bibitem{gravel2011demographic}
\begin{barticle}[author]
\bauthor{\bsnm{Gravel},~\bfnm{Simon}\binits{S.}},
  \bauthor{\bsnm{Henn},~\bfnm{Brenna~M}\binits{B.~M.}},
  \bauthor{\bsnm{Gutenkunst},~\bfnm{Ryan~N}\binits{R.~N.}},
  \bauthor{\bsnm{Indap},~\bfnm{Amit~R}\binits{A.~R.}},
  \bauthor{\bsnm{Marth},~\bfnm{Gabor~T}\binits{G.~T.}},
  \bauthor{\bsnm{Clark},~\bfnm{Andrew~G}\binits{A.~G.}},
  \bauthor{\bsnm{Yu},~\bfnm{Fuli}\binits{F.}},
  \bauthor{\bsnm{Gibbs},~\bfnm{Richard~A}\binits{R.~A.}},
  \bauthor{\bsnm{Bustamante},~\bfnm{Carlos~D}\binits{C.~D.}},
  \bauthor{\bsnm{Altshuler},~\bfnm{David~L}\binits{D.~L.}} \AND \betal{et~al.}
(\byear{2011}).
\btitle{Demographic history and rare allele sharing among human populations}.
\bjournal{Proceedings of the National Academy of Sciences}
\bvolume{108}
\bpages{11983--11988}.
\end{barticle}
\endbibitem

\bibitem{griffiths1998age}
\begin{barticle}[author]
\bauthor{\bsnm{Griffiths},~\bfnm{R.C.}\binits{R.}} \AND
  \bauthor{\bsnm{Tavar{\'e}},~\bfnm{Simon}\binits{S.}}
(\byear{1998}).
\btitle{The age of a mutation in a general coalescent tree}.
\bjournal{Communications in Statistics. Stochastic Models}
\bvolume{14}
\bpages{273--295}.
\end{barticle}
\endbibitem

\bibitem{gutenkunst2009inferring}
\begin{barticle}[author]
\bauthor{\bsnm{Gutenkunst},~\bfnm{Ryan~N.}\binits{R.~N.}},
  \bauthor{\bsnm{Hernandez},~\bfnm{Ryan~D.}\binits{R.~D.}},
  \bauthor{\bsnm{Williamson},~\bfnm{Scott~H.}\binits{S.~H.}} \AND
  \bauthor{\bsnm{Bustamante},~\bfnm{Carlos~D.}\binits{C.~D.}}
(\byear{2009}).
\btitle{Inferring the Joint Demographic History of Multiple Populations from
  Multidimensional {SNP} Frequency Data}.
\bjournal{PLoS Genetics}
\bvolume{5}
\bpages{e1000695}.
\end{barticle}
\endbibitem

\bibitem{Kim201526}
\begin{barticle}[author]
\bauthor{\bsnm{Kim},~\bfnm{Junhyong}\binits{J.}},
  \bauthor{\bsnm{Mossel},~\bfnm{Elchanan}\binits{E.}},
  \bauthor{\bsnm{R\'acz},~\bfnm{Mikl\'os~Z.}\binits{M.~Z.}} \AND
  \bauthor{\bsnm{Ross},~\bfnm{Nathan}\binits{N.}}
(\byear{2015}).
\btitle{Can one hear the shape of a population history?}
\bjournal{Theoretical Population Biology}
\bvolume{100}
\bpages{26---38}.
\end{barticle}
\endbibitem

\bibitem{kimura1969number}
\begin{barticle}[author]
\bauthor{\bsnm{Kimura},~\bfnm{Motoo}\binits{M.}}
(\byear{1969}).
\btitle{The number of heterozygous nucleotide sites maintained in a finite
  population due to steady flux of mutations}.
\bjournal{Genetics}
\bvolume{61}
\bpages{893}.
\end{barticle}
\endbibitem

\bibitem{kingman1982coalescent}
\begin{barticle}[author]
\bauthor{\bsnm{Kingman},~\bfnm{J.~F.~C.}\binits{J.~F.~C.}}
(\byear{1982}).
\btitle{The coalescent}.
\bjournal{Stoch. Process. Appl.}
\bvolume{13}
\bpages{235-248}.
\end{barticle}
\endbibitem

\bibitem{kingman1982exchangeability}
\begin{bincollection}[author]
\bauthor{\bsnm{Kingman},~\bfnm{J.~F.~C.}\binits{J.~F.~C.}}
(\byear{1982}).
\btitle{Exchangeability and the evolution of large populations}.
In \bbooktitle{Exchangeability in Probability and Statistics}
(\beditor{\bfnm{G.}\binits{G.}~\bsnm{Koch}} \AND
  \beditor{\bfnm{F.}\binits{F.}~\bsnm{Spizzichino}}, eds.)
\bpages{97--112}.
\bpublisher{North-Holland Publishing Company}.
\end{bincollection}
\endbibitem

\bibitem{kingman1982genealogy}
\begin{barticle}[author]
\bauthor{\bsnm{Kingman},~\bfnm{J.~F.~C.}\binits{J.~F.~C.}}
(\byear{1982}).
\btitle{On the genealogy of large populations}.
\bjournal{J. Appl. Prob.}
\bvolume{19A}
\bpages{27-43}.
\end{barticle}
\endbibitem

\bibitem{li2011inference}
\begin{barticle}[author]
\bauthor{\bsnm{Li},~\bfnm{Heng}\binits{H.}} \AND
  \bauthor{\bsnm{Durbin},~\bfnm{Richard}\binits{R.}}
(\byear{2011}).
\btitle{Inference of human population history from individual whole-genome
  sequences}.
\bjournal{Nature}
\bvolume{475}
\bpages{493--496}.
\end{barticle}
\endbibitem

\bibitem{massart2007concentration}
\begin{bbook}[author]
\bauthor{\bsnm{Massart},~\bfnm{Pascal}\binits{P.}}
(\byear{2007}).
\btitle{Concentration inequalities and model selection}
\bvolume{1896}.
\bpublisher{Berlin: Springer}.
\end{bbook}
\endbibitem

\bibitem{myers2008can}
\begin{barticle}[author]
\bauthor{\bsnm{Myers},~\bfnm{Simon}\binits{S.}},
  \bauthor{\bsnm{Fefferman},~\bfnm{Charles}\binits{C.}} \AND
  \bauthor{\bsnm{Patterson},~\bfnm{Nick}\binits{N.}}
(\byear{2008}).
\btitle{Can one learn history from the allelic spectrum?}
\bjournal{Theor. Popul. Biol.}
\bvolume{73}
\bpages{342--348}.
\end{barticle}
\endbibitem

\bibitem{nelson2012abundance}
\begin{barticle}[author]
\bauthor{\bsnm{Nelson},~\bfnm{Matthew~R}\binits{M.~R.}},
  \bauthor{\bsnm{Wegmann},~\bfnm{Daniel}\binits{D.}},
  \bauthor{\bsnm{Ehm},~\bfnm{Margaret~G}\binits{M.~G.}},
  \bauthor{\bsnm{Kessner},~\bfnm{Darren}\binits{D.}},
  \bauthor{\bsnm{Jean},~\bfnm{Pamela~St}\binits{P.~S.}},
  \bauthor{\bsnm{Verzilli},~\bfnm{Claudio}\binits{C.}},
  \bauthor{\bsnm{Shen},~\bfnm{Judong}\binits{J.}},
  \bauthor{\bsnm{Tang},~\bfnm{Zhengzheng}\binits{Z.}},
  \bauthor{\bsnm{Bacanu},~\bfnm{Silviu-Alin}\binits{S.-A.}},
  \bauthor{\bsnm{Fraser},~\bfnm{Dana}\binits{D.}} \AND \betal{et~al.}
(\byear{2012}).
\btitle{An Abundance of Rare Functional Variants in 202 Drug Target Genes
  Sequenced in 14,002 People}.
\bjournal{Science}
\bvolume{337}
\bpages{100--104}.
\end{barticle}
\endbibitem

\bibitem{nielsen2000estimation}
\begin{barticle}[author]
\bauthor{\bsnm{Nielsen},~\bfnm{Rasmus}\binits{R.}}
(\byear{2000}).
\btitle{Estimation of population parameters and recombination rates from single
  nucleotide polymorphisms}.
\bjournal{Genetics}
\bvolume{154}
\bpages{931--942}.
\end{barticle}
\endbibitem

\bibitem{paul2011accurate}
\begin{barticle}[author]
\bauthor{\bsnm{Paul},~\bfnm{Joshua~S.}\binits{J.~S.}},
  \bauthor{\bsnm{Steinr\"ucken},~\bfnm{Matthias}\binits{M.}} \AND
  \bauthor{\bsnm{Song},~\bfnm{Yun~S.}\binits{Y.~S.}}
(\byear{2011}).
\btitle{An accurate sequentially {M}arkov conditional sampling distribution for
  the coalescent with recombination}.
\bjournal{Genetics}
\bvolume{187}
\bpages{1115--1128}.
\bnote{(PMC3070520)}.
\end{barticle}
\endbibitem

\bibitem{polanski2003note}
\begin{barticle}[author]
\bauthor{\bsnm{Polanski},~\bfnm{Andrzej}\binits{A.}},
  \bauthor{\bsnm{Bobrowski},~\bfnm{Adam}\binits{A.}} \AND
  \bauthor{\bsnm{Kimmel},~\bfnm{Marek}\binits{M.}}
(\byear{2003}).
\btitle{A note on distributions of times to coalescence, under time-dependent
  population size}.
\bjournal{Theoretical Population Biology}
\bvolume{63}
\bpages{33--40}.
\end{barticle}
\endbibitem

\bibitem{polanski2003new}
\begin{barticle}[author]
\bauthor{\bsnm{Polanski},~\bfnm{Andrzej}\binits{A.}} \AND
  \bauthor{\bsnm{Kimmel},~\bfnm{Marek}\binits{M.}}
(\byear{2003}).
\btitle{New Explicit Expressions for Relative Frequencies of Single-Nucleotide
  Polymorphisms With Application to Statistical Inference on Population
  Growth}.
\bjournal{Genetics}
\bvolume{165}
\bpages{427--436}.
\end{barticle}
\endbibitem

\bibitem{rasmussen2014genome}
\begin{barticle}[author]
\bauthor{\bsnm{Rasmussen},~\bfnm{Matthew~D}\binits{M.~D.}},
  \bauthor{\bsnm{Hubisz},~\bfnm{Melissa~J}\binits{M.~J.}},
  \bauthor{\bsnm{Gronau},~\bfnm{Ilan}\binits{I.}} \AND
  \bauthor{\bsnm{Siepel},~\bfnm{Adam}\binits{A.}}
(\byear{2014}).
\btitle{Genome-wide inference of ancestral recombination graphs}.
\bjournal{PLoS Genetics}
\bvolume{10}
\bpages{e1004342}.
\end{barticle}
\endbibitem

\bibitem{schiffels2014inferring}
\begin{barticle}[author]
\bauthor{\bsnm{Schiffels},~\bfnm{Stephan}\binits{S.}} \AND
  \bauthor{\bsnm{Durbin},~\bfnm{Richard}\binits{R.}}
(\byear{2014}).
\btitle{Inferring human population size and separation history from multiple
  genome sequences}.
\bjournal{Nature Genetics}
\bvolume{46}
\bpages{919-925}.
\end{barticle}
\endbibitem

\bibitem{sheehan2013estimating}
\begin{barticle}[author]
\bauthor{\bsnm{Sheehan},~\bfnm{Sara}\binits{S.}},
  \bauthor{\bsnm{Harris},~\bfnm{Kelley}\binits{K.}} \AND
  \bauthor{\bsnm{Song},~\bfnm{Yun~S}\binits{Y.~S.}}
(\byear{2013}).
\btitle{Estimating variable effective population sizes from multiple genomes: A
  sequentially {M}arkov conditional sampling distribution approach}.
\bjournal{Genetics}
\bvolume{194}
\bpages{647--662}.
\bnote{(PMC3697970)}.
\end{barticle}
\endbibitem

\bibitem{steinruecken2013sequentially}
\begin{barticle}[author]
\bauthor{\bsnm{Steinr{\"u}cken},~\bfnm{Matthias}\binits{M.}},
  \bauthor{\bsnm{Paul},~\bfnm{Joshua~S}\binits{J.~S.}} \AND
  \bauthor{\bsnm{Song},~\bfnm{Yun~S}\binits{Y.~S.}}
(\byear{2013}).
\btitle{A sequentially {M}arkov conditional sampling distribution for
  structured populations with migration and recombination}.
\bjournal{Theor. Popul. Biol.}
\bvolume{87}
\bpages{51--61}.
\bnote{(PMC3532580)}.
\end{barticle}
\endbibitem

\bibitem{tennessen2012evolution}
\begin{barticle}[author]
\bauthor{\bsnm{Tennessen},~\bfnm{Jacob~A}\binits{J.~A.}},
  \bauthor{\bsnm{Bigham},~\bfnm{Abigail~W}\binits{A.~W.}},
  \bauthor{\bsnm{O'Connor},~\bfnm{Timothy~D}\binits{T.~D.}},
  \bauthor{\bsnm{Fu},~\bfnm{Wenqing}\binits{W.}},
  \bauthor{\bsnm{Kenny},~\bfnm{Eimear~E}\binits{E.~E.}},
  \bauthor{\bsnm{Gravel},~\bfnm{Simon}\binits{S.}},
  \bauthor{\bsnm{McGee},~\bfnm{Sean}\binits{S.}},
  \bauthor{\bsnm{Do},~\bfnm{Ron}\binits{R.}},
  \bauthor{\bsnm{Liu},~\bfnm{Xiaoming}\binits{X.}},
  \bauthor{\bsnm{Jun},~\bfnm{Goo}\binits{G.}} \AND \betal{et~al.}
(\byear{2012}).
\btitle{Evolution and Functional Impact of Rare Coding Variation from Deep
  Sequencing of Human Exomes}.
\bjournal{Science}
\bvolume{337}
\bpages{64--69}.
\end{barticle}
\endbibitem

\bibitem{tsybakov}
\begin{bbook}[author]
\bauthor{\bsnm{Tsybakov},~\bfnm{A~B}\binits{A.~B.}}
(\byear{2009}).
\btitle{Introduction to Nonparametric Estimation}.
\bpublisher{Springer Science+Business Media, New York}.
\end{bbook}
\endbibitem

\bibitem{yu1997assouad}
\begin{bincollection}[author]
\bauthor{\bsnm{Yu},~\bfnm{Bin}\binits{B.}}
(\byear{1997}).
\btitle{Assouad, Fano, and Le Cam}.
In \bbooktitle{Festschrift for Lucien Le Cam}
(\beditor{\bfnm{Grace L.~Yang}\binits{G.~L.~Y.}~\bsnm{David~Pollard}\bsuffix{
  Erik~Torgersen}}, ed.)
\bpages{423--435}.
\bpublisher{Springer}.
\end{bincollection}
\endbibitem

\bibitem{exac}
\begin{bmisc}[author]
\btitle{{The Exome Aggregation Consortium}}.
\bhowpublished{\url{http://exac.broadinstitute.org/}}.
\end{bmisc}
\endbibitem

\bibitem{uk10k}
\begin{bmisc}[author]
\btitle{{The UK10K Project}: Rare Genetic Variants in Health and Disease}.
\bhowpublished{\url{http://www.uk10k.org/}}.
\end{bmisc}
\endbibitem

\end{thebibliography}
\end{document}